\documentclass[submission,copyright,creativecommons]{eptcs}
\providecommand{\event}{AiML 2026} 

\usepackage{aiml26}

\usepackage{iftex}

\usepackage{tikz-cd}
\usepackage{tikz}
\usepackage{amsmath,amssymb,xcolor,amsthm}
\usetikzlibrary{decorations.pathmorphing, decorations.markings}
\usepackage[all]{xy}
\usepackage{ifthen,xspace}
\usepackage{graphicx,color}
\usepackage{MnSymbol}

\ifpdf
  \usepackage{underscore}         
  \usepackage[T1]{fontenc}        
\else
  \usepackage{breakurl}           
\fi

\title{Non-finite Axiomatizability of Generalized Medvedev Logics}
\author{Han Xiao
\institute{The Tsinghua-UvA JRC for Logic\\ Department of Philosophy\\  Tsinghua University\\ Beijing, China}
\email{han\_xiao@mail.tsinghua.edu.cn}
}

\newcommand{\titlerunning}{Generalized Medvedev Logics}
\newcommand{\authorrunning}{H. Xiao}

\hypersetup{
  bookmarksnumbered,
  pdftitle    = {\titlerunning},
  pdfauthor   = {\authorrunning},
  pdfsubject  = {EPTCS},               
  pdfkeywords = {keyword1, keyword2} 
}

\begin{document}
\maketitle

\begin{abstract}
We introduce a generalized form of Medvedev logics obtained by removing the greatest element from finite products of rooted Kripke frames with a top. We show that, before removing the top, the intermediate logic characterized by such finite products is exactly $\mathsf{KC}$. Classical Medvedev logic is characterized by topless products of $2$-chains, and a theorem of Maksimova, Skvortsov and Shehtman establishes that it is not finitely axiomatizable. Motivated by this result, Nick Bezhanishvili conjectured that non-finite axiomatizability extends to topless products of arbitrary finite chains and, more generally, to topless products of finite rooted frames with a top. We prove that every such generalized Medvedev logic is not finitely axiomatizable, thereby settling both conjectures in the affirmative. In 2003, van Benthem, Guram Bezhanishvili, and Gehrke introduced $\mathsf{Cheq}$, the logic of chequered sets, and we show that whenever $\mathsf{Cheq}$ is a sublogic of a generalized Medvedev logic, the latter is not finitely axiomatizable over $\mathsf{Cheq}$. Finally, we investigate the order structure of generalized Medvedev logics. We prove that there are at least countably many distinct generalized Medvedev logics and that no least such logic exists. These results extend the classical theory of Medvedev logic and clarify the behaviour of intermediate logics generated by topless product constructions.
\end{abstract}

\section{Introduction}
\subsection{General motivation}
In 1962, Medvedev introduced the logic of finite problems $\mathsf{Med}$ \cite{Med62} as a formalization of Kolmogorov's \cite{Kol32} interpretation of intuitionistic logic as a calculus of problems.
One of the most important results about \emph{Medvedev logic} $\mathsf{Med}$ was obtained by Maksimova, Skvortsov, and Shehtman, who proved in 1979 that $\mathsf{Med}$ is not finitely axiomatizable \cite[Corollary 5]{MSS79}.

Beyond its role in the theory of intermediate logics, $\mathsf{Med}$ also arises in the study of \emph{modal logics of forcing}. This line of research, initiated by Hamkins and L{\"o}we \cite{HL08}, investigates the modal logic of the class of models of set theory with the forcing relation. In particular, determining the \emph{modal logics of c.c.c.\ forcing} has become a central problem in this area. A notable connection with Medvedev logic was established by Hamkins, Leibman, and L{\"o}we, who proved that the modal logic of $\omega_1$-preserving forcing is contained in $\mathsf{S4.tBA}$, the least modal companion of $\mathsf{Med}$ \cite[Theorem~36]{HLL15}. They further conjectured that $\mathsf{S4.tBA}$ also provides an upper bound for the modal logic of c.c.c.\ forcing.

A useful structural perspective on Medvedev logic is that it is characterized by \emph{topless products} of the $2$-chain. More precisely, the Medvedev frame $\mathbf{P}_0(n)$ is obtained from the $n$-fold product of the $2$-chain by removing its greatest element. 
From this perspective, it is natural to ask whether these constructions lead to new intermediate logics with similar structural properties. 
In particular, one may replace the $2$-chain by more general finite frames and consider the topless products obtained by removing the greatest element from their finite products. This idea was proposed by Nick Bezhanishvili, who asked whether the classical non-finite axiomatizability result for $\mathsf{Med}$ extends to such generalized constructions. 
More specifically, he conjectured that non-finite axiomatizability should already hold for topless products of arbitrary  non-singleton finite chains and, more generally, for topless products of arbitrary non-singleton finite rooted frames with a top.
These considerations lead naturally to a broader class of intermediate logics that we call \emph{generalized Medvedev logics}. 
Studying these logics provides a systematic way to understand how structural features of Kripke frames-such as branching and product constructions-affect the axiomatizability. 
 
\subsection{Outline of the paper}
We first study the logics characterized by finite 
products of a non-singleton finite rooted frame 
with a top. We prove that the intermediate 
logic characterized by these products is always the weak excluded middle logic $\mathsf{KC}$ (Theorem~\ref{thm:pifkc}). 

In contrast, once the top element is removed, the resulting generalized Medvedev logics exhibit non-finite axiomatizability. 
As a first structural observation, we show that the classical Medvedev logic is the largest 
generalized Medvedev logic (Lemma~\ref{thm:medmaxpift}). 
The main result of the paper proves that every generalized Medvedev logic is not finitely axiomatizable, thereby confirming conjectures of Bezhanishvili for non-trivial finite distributive lattices (Theorem~\ref{thm:chaincasemed}) and, more generally, for arbitrary non-singleton finite rooted frames with a top (Theorem~\ref{thm:generalcasenick}).

Beyond this central result, we further analyze the relationship between 
generalized Medvedev logics and other well-studied intermediate logics. In particular, we extend Fontaine's theorem by proving that whenever $\mathsf{Cheq}$ is a sublogic of a generalized Medvedev logic, the latter is not finitely axiomatizable over $\mathsf{Cheq}$ (Theorem~\ref{pro:genfontaine}).

Finally, we investigate the order structure of generalized Medvedev logics, showing that there are at least countably many pairwise distinct such 
logics (Theorem~\ref{thm:countablygen}) and that no least generalized Medvedev logic exists (Theorem~\ref{thm:noleastgen}).

\section{Definitions}
\label{sec:def}

Throughout this paper, we assume familiarity with the basic theory of intermediate logics. 
We write $\mathsf{IPC}$ for \emph{intuitionistic propositional calculus}. 
By an (intuitionistic Kripke) frame we mean a poset. We use the terms `frame' and `poset' interchangeably and all posets considered are finite.

Given a frame $\mathbf{F}$, we write $\mathsf{Log}(\mathbf{F})$ for the propositional logic consisting of all formulas valid on the frame $\mathbf{F}$; if
$\mathsf{L} \subseteq \mathsf{Log}(\mathbf{F})$,
we call $\mathbf{F}$ an \emph{$\mathsf{L}$-frame}. If $\mathcal{C}$ is a class of frames, we write $\mathsf{Log}(\mathcal{C})$ for the intermediate logic consisting of all formulas valid on every frame of $\mathcal{C}$; in this case, we say that $\mathcal{C}$ \emph{characterizes} the logic.
\subsection{Posets as Kripke frames} 

Let $\mathbf{F} = \langle F,\leq \rangle$ be a finite poset; we write $x<y$ if $x\leq y$ and $x\neq y$. An element $x \in F$ is called \emph{maximal} if there is no $y \in F$ such that $x < y$; it is called a \emph{top} (or \emph{greatest element}) if everything is below it. A top is necessarily unique.
We say that $\mathbf{F}$ is a frame \emph{with a top} if it contains a top.
If $x<y$, we call $y$ an \emph{immediate successor} of $x$ if there is no $z$ such that $x<z<y$. The \emph{branching degree} of $x$, denoted by $b(x)$, is defined as the number of immediate successors of $x$. 

The notation $x{\uparrow}$ (resp.\ $x{\downarrow}$) denotes the \emph{upset} $\{y\in F : x\leq y\}$ (resp.\ the \emph{downset} $\{y\in F : x\geq y\}$). A subframe $\mathbf{S}=\langle S, \leq \rangle$ of $\mathbf{F} = \langle F,\leq \rangle$ is called a \emph{generated subframe} if $S$ is upward closed in $F$; we say that $\mathbf{S}$ is \emph{generated by $x$} if $S = x{\uparrow}$. If $\mathbf{F}$ is generated by $x$, then $\mathbf{F}$ is called \emph{rooted} and $x$ the \emph{root} of $\mathbf{F}$. For $n \ge 1$, let $[n] := \{1,2,\dots,n\}$. The \emph{$n$-chain} is the frame $\mathbf{C}(n) := \langle [n], \leq \rangle$, where $\leq$ is the usual order on $[n]$. 
The \emph{depth} of a frame $\mathbf{F}$ is the length of the longest 
chain in $\mathbf{F}$. 
For $x \in F$, the \emph{depth of $x$}, denoted by $d(x)$, is the depth of the subframe 
generated by $x$.

Let $\mathbf{F}=\langle F,\le_0\rangle$ and 
$\mathbf{G}=\langle G,\le_1\rangle$ be posets. 
Their \emph{product} is the poset 
$\mathbf{F}\times\mathbf{G}
:= \langle F\times G,\le\rangle$, 
where the order $\le$ is defined coordinatewise by $(x,y)\le (x',y')$ if and only if $x\le_0 x'$ and $y\le_1 y'$.
For a poset $\mathbf{F}$, we write 
$\mathbf{F}^2 := \mathbf{F}\times\mathbf{F}$ and define inductively
$\mathbf{F}^n := \mathbf{F}^{\,n-1}\times\mathbf{F}$. Now let $\mathbf{F}$ be a finite frame with a top.
Its \emph{topless frame} $\mathbf{F}^{-}$ is the frame obtained from 
$\mathbf{F}$ by removing its top element.
For $n\ge1$, the \emph{$n$-fold topless product} of $\mathbf{F}$ is $\mathbf{F}^{n-} := (\mathbf{F}^n)^{-}$, that is, the topless frame of the product $\mathbf{F}^n$.

We write $\mathbf{F}\oplus\mathbf{G}$ for their \emph{linear sum}, 
obtained by placing all elements of $\mathbf{G}$ strictly above all 
elements of $\mathbf{F}$. 
We define the iterated finite sum by recursion via $\mathbf{F}\times 1 := \mathbf{F}$ and $\mathbf{F}\times (n+1) := (\mathbf{F}\times n)\oplus \mathbf{F}$. 
The \emph{vertical sum} of $\mathbf{F}$ and $\mathbf{G}$ is obtained from the linear sum of $\mathbf{F}$ and $\mathbf{G}$ by identifying the top of $\mathbf{F}$ with the root of $\mathbf{G}$ (provided they exist).

\subsection{Morphisms} If $\mathbf{F}$ and $\mathbf{G}$ are posets and $f:F\to G$ is a map, we call $f$ a \emph{p-morphism} if for all $x,y\in F$, if $x\leq y$, then $f(x)\leq f(y)$ and
if $f(x)\leq u$, then there is some $y\geq x$ such that $f(y) = u$. If $f$ is onto, then we call $\mathbf{G}$ a \emph{p-morphic image} of $\mathbf{F}$.

\begin{theorem}[Folklore]
If $\mathbf{G}$ is a $p$-morphic image of $\mathbf{F}$, then $\mathbf{F}\vDash \varphi \text{ implies } \mathbf{G}\vDash \varphi$, for any formula $\varphi$.
\end{theorem}
 
\begin{theorem}[Jankov-de Jongh Theorem]
	Given a finite rooted frame $\mathbf{F}$, there exists a formula $\chi(\mathbf{F})$ such that for every frame $\mathbf{G}$, $$\mathbf{G} \nvDash \chi(\mathbf{F})
\quad\text{iff}\quad
\mathbf{F}
\text{ is a $p$-morphic image of a generated subframe of }
\mathbf{G}.$$The formula $\chi(\mathbf{F})$ is called the Jankov-de Jongh formula of $\mathbf{F}$.
\end{theorem}

\subsection{Medvedev logic}
In 1932, Kolmogorov \cite{Kol32} suggested a constructive interpretation of intuitionistic logic as a calculus of problems. To make precise the description proposed by Kolmogorov, Medvedev established the foundational framework for the logic of finite problems in \cite{Med62}. We recall the definition of the \emph{Medvedev frame} and the corresponding \emph{Medvedev logic}.
\begin{definition}
For $n \ge 1$, the \emph{Medvedev frame on $n$ atoms} $\mathbf{P}_0(n)$ is defined by
$\mathbf{P}_0(n) 
:= \bigl\langle 
\{ X \subseteq [n] : X \neq \varnothing \}, 
\supseteq 
\bigr\rangle$.
The \emph{Medvedev logic} $\mathsf{Med}$ is the intermediate logic of all Medvedev frames, that is, $\mathsf{Med}:= \mathsf{Log}\bigl(\{ \mathbf{P}_0(n) : n \in \omega \}\bigr)$.
\end{definition}
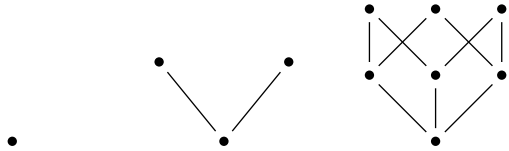
\begin{figure}[htbp]
\centering
\begin{tikzpicture}[
scale=0.35,
every node/.style={
    circle,
    fill=black,
    minimum size=3.5pt,
    inner sep=0pt
},
edge/.style={
    line width=0.5pt,
    shorten >=3pt,
    shorten <=3pt
}
]

\begin{scope}[xshift=0cm] 
    \node (a1) at (0,1) {};
    
\end{scope}

\begin{scope}[xshift=8cm] 
    \node (b1) at (0,1) {};
    \node (b2) at (-2.45,4) {};
    \node (b3) at (2.45,4) {};

    \draw[edge] (b1) -- (b2);
    \draw[edge] (b1) -- (b3);
\end{scope}

\begin{scope}[xshift=16cm] 
    \node (c1) at (0,1) {};
    \node (c2) at (-2.5,3.5) {};
    \node (c3) at (0,3.5) {};
    \node (c4) at (2.5,3.5) {};
    \node (c5) at (-2.5,6) {};
    \node (c6) at (0,6) {};
    \node (c7) at (2.5,6) {};

    \draw[edge] (c1) -- (c2);
    \draw[edge] (c1) -- (c3);
    \draw[edge] (c1) -- (c4);
    \draw[edge] (c2) -- (c5);
    \draw[edge] (c2) -- (c6);
    \draw[edge] (c3) -- (c5);
    \draw[edge] (c3) -- (c7);
    \draw[edge] (c4) -- (c6);
    \draw[edge] (c4) -- (c7);
\end{scope}

\end{tikzpicture}

\caption{Medvedev frames.}
\label{fig:med}
\end{figure}
\begin{lemma}[Maksimova, Skvortsov, \& Shehtman; cf.\ {\cite[Lemma 4]{MSS79}}]
If $\mathbf{F}$ is a finite rooted frame with a top, then $\mathbf{F}$ is a $\mathsf{Med}$-frame.		
\label{lem:topmedframe}
\end{lemma}
Inheriting such a tradition on the frame for logic of finite problems, the order in our \emph{finite Boolean algebra} here is reverse inclusion $\supseteq$. 
\begin{definition}
For $n \ge 1$, the \emph{Boolean algebra on $n$ generators} $\mathbf{P}(n)$ is defined by $\mathbf{P}(n):=\bigl\langle 
\mathcal{P}([n]), \supseteq \bigr\rangle$, where $\mathcal{P}([n])$ denotes the powerset of $[n]$.
\end{definition}

\subsection{The logic $\mathsf{Cheq}$}

The modal logic of chequered subsets was introduced by 
van Benthem, Guram Bezhanishvili, and Gehrke in 2003, who asked 
whether it is finitely axiomatizable \cite[in particular, p.~343]{vBBG03}. 
This logic plays an important role in the study of modal logics of 
various topological spaces. Litak \cite{Lit04} later denoted its 
intermediate companion by $\mathsf{Cheq}$.
Let $\mathbf{V}$ be the \emph{two-fork}, that is, the poset consisting 
of three elements $\{o,p,q\}$ such that $o$ is the root and $p$ and $q$ 
are incomparable maximal elements. Define recursively 
$\mathbf{V}_1 := \mathbf{V}$ and 
$\mathbf{V}_{n+1} := \mathbf{V}_n \times \mathbf{V}$,
where $\times$ denotes the usual product of posets. 
The logic $\mathsf{Cheq}$ is the intermediate logic characterized by 
the class $\{\mathbf{V}_n : n \in \omega\}$.

\section{Products and generalized Medvedev logics}

Observe that $\mathbf{P}(n)$ is isomorphic to $\mathbf{C}(2)^n$, the $n$-fold product 
of the $2$-chain, and that the Medvedev frame 
$\mathbf{P}_0(n)$ is obtained from $\mathbf{C}(2)^n$ by removing its 
greatest element. 
Thus, Medvedev logic can be viewed as the logic 
characterized by the family of frames obtained as topless products of 
the $2$-chain.

This structural description suggests a natural generalization. 
Instead of taking products of the $2$-chain, one may consider 
finite products of more general frames and remove their top 
elements. In particular, one may replace the $2$-chain by an 
arbitrary finite chain, or more generally by an arbitrary 
non-singleton finite rooted frame with a top. Motivated by this observation, Bezhanishvili proposed two conjectures concerning the logics of such topless products. These conjectures predict that the non-finite axiomatizability known for Medvedev logic should persist for these 
generalized constructions. 

In the remainder of this section we introduce the corresponding 
class of logics, which we call generalized Medvedev logics, 
and investigate their properties.

\begin{proposition}
	The Medvedev frame $\mathbf{P}_0(n)$ is isomorphic to $\mathbf{C}(2)^{n-}$.
\end{proposition}
\begin{proof}
Define a map $f: \mathbf{P}_0(n)\to \mathbf{C}(2)^{n-}$ by sending each non-empty subset $X\subseteq [n]$ to the tuple $f(X)=(x_1,\ldots,x_n)$, where $x_i=2$ if $i\notin X$ and $x_i=1$ otherwise. This map is clearly bijective. Moreover, $X \supseteq Y$ if and only if $f(X) \leq f(Y)$, 
so $f$ is an isomorphism of posets.
\end{proof}

\begin{definition}
Let $\mathbf{F}$ be a non-singleton finite rooted frame with a top. We define
$\mathsf{P}_{\mathbf{F}}
:=
\mathsf{Log}\bigl(\{ \mathbf{F}^n : n \in \omega\}\bigr)$ and $\mathsf{TLP}_{\mathbf{F}}
:=
\mathsf{Log}\bigl(\{ \mathbf{F}^{n-} : n \in \omega \}\bigr).$ The logic $\mathsf{TLP}_{\mathbf{F}}$ is called the 
\emph{generalized Medvedev logic over $\mathbf{F}$}.
\end{definition}
Throughout the rest of this paper, unless otherwise stated, 
$\mathbf{F}$ is assumed to be a non-singleton finite rooted frame with a top. The following two results, Theorem~\ref{thm:pifkc} and 
Lemma~\ref{thm:medmaxpift}, provide an initial description 
of the logics $\mathsf{P}_{\mathbf{F}}$ and 
$\mathsf{TLP}_{\mathbf{F}}$. We begin with the case of products. 
The logic characterized by the finite products of $\mathbf{F}$ 
turns out to be the weak excluded middle logic $\mathsf{KC} = \mathsf{IPC} + (\neg p \vee \neg\neg p)$, also known as \emph{Jankov logic}. Recall that $\mathsf{KC}$ is the logic of all finite rooted posets that have a top \cite[Theorem~2.13]{BB22} and, in particular, is characterized by the products of the $2$-chain. The following theorem shows that this result persists for 
finite products of an arbitrary finite rooted frame with a top.  Theorem~\ref{thm:pifkc} was conjectured by Bezhanishvili and the proof presented below is due to the author of the present paper.

\begin{theorem}
$\mathsf{P}_\mathbf{F}=\mathsf{KC}$.
\label{thm:pifkc}	
\end{theorem}
\begin{proof}
Since $\mathsf{KC}$ is the intermediate logic of all finite rooted frames with a top and each $\mathbf{F}^n$ has a top, we immediately obtain $\mathsf{P}_\mathbf{F}\supseteq\mathsf{KC}$. Let $t$ denote the top of $\mathbf{F}$. Define a map $f_1: \mathbf{F} \to \mathbf{C}(2)$ by $$f_1(x)=\begin{cases}
	2, & \text{if } x=t,\\		
	1, & \text{otherwise}.\\
\end{cases}$$ 

A straightforward verification shows that $f_1$ is a $p$-morphism. For each $n\ge 1$, define $f_n:\mathbf{F}^n\to\mathbf{C}(2)^n$ by $f_n\bigl((x_1,x_2,\ldots,x_n)\bigr)=\bigl(f_1(x_1),f_1(x_2),\ldots,f_1(x_n)\bigr).$ It is easy to check that $f_n$ is a $p$-morphism from $\mathbf{F}^n$ onto $\mathbf{C}(2)^n$. Consequently, $\mathsf{Log}(\mathbf{F}^n)\subseteq
\mathsf{Log}(\mathbf{C}(2)^n)$ and thus $\mathsf{P}_\mathbf{F}=\bigcap_{n\in \omega} \mathsf{Log}(\mathbf{F}^n) \subseteq \bigcap_{n\in \omega}\mathsf{Log}\bigl(\mathbf{C}(2)^n\bigr)$. Since $\mathsf{KC} = \bigcap_{n\in \omega}\mathsf{Log}\bigl(\mathbf{C}(2)^n\bigr)$ (see, e.g., \cite[Corollary 1]{MSS79}), we obtain $\mathsf{P}_\mathbf{F} \subseteq \mathsf{KC}$. Therefore, $\mathsf{P}_{\mathbf{F}} = \mathsf{KC}$.
\end{proof}

From the perspective of topless products $\mathbf{F}^{n-}$, the classical Medvedev logic is the largest generalized Medvedev logic with respect to inclusion.

\begin{lemma}
$\mathsf{TLP}_\mathbf{F}\subseteq\mathsf{Med}$.
\label{thm:medmaxpift}
\end{lemma}
\begin{proof}
We use the maps $f_n$ constructed in the proof of 
Theorem~\ref{thm:pifkc}. Recall that $f_n$ is a $p$-morphism from 
$\mathbf{F}^n$ onto $\mathbf{C}(2)^n$ and the Medvedev frame $\mathbf{P}_0(n)$ is the topless frame of $\mathbf{C}(2)^n$. Define $f_n^0:\mathbf{F}^{n-}\to \mathbf{P}_0(n)$ by restriction $f_n^0 := f_n\!\upharpoonright\,\mathbf{F}^{n-}$. The only element of $\mathbf{F}^n$ mapped to the top of 
$\mathbf{C}(2)^n$ by $f_n$ is the top of $\mathbf{F}^n$. 
Hence, the restriction $f_n^0$ remains a $p$-morphism. Thus, $\mathsf{Log}(\mathbf{F}^{n-})\subseteq\mathsf{Log}(\mathbf{P}_0(n))$. Taking intersections over all $n$, we obtain $\mathsf{TLP}_{\mathbf{F}}=\bigcap_{n\in\omega} \mathsf{Log}(\mathbf{F}^{n-})\subseteq\bigcap_{n\in \omega} \mathsf{Log}(\mathbf{P}_0(n))=\mathsf{Med}$. 
\end{proof}

This paper resolves two conjectures concerning the finite axiomatizability of generalized Medvedev logics, which were suggested by Nick Bezhanishvili in personal communication.

\begin{conjecture}[Bezhanishvili's First Conjecture]
If $\mathbf{F} = \mathbf{C}(n)$ is a finite chain with $n \ge 2$, 
then $\mathsf{TLP}_{\mathbf{F}}$ is not finitely axiomatizable.
\label{con:chaincase}
\end{conjecture}

\begin{conjecture}[Bezhanishvili's Second Conjecture]
$\mathsf{TLP}_{\mathbf{F}}$ is not finitely axiomatizable.
\label{con:generalcase}
\end{conjecture}

These conjectures extend the classical result of 
Maksimova, Skvortsov, and Shehtman 
\cite[Corollary~5]{MSS79}, which shows that $\mathsf{Med}$ is not finitely axiomatizable.

\section{Non-finite axiomatization of $\mathsf{TLP}_{\mathbf{F}}$}
\label{sec:nofaofgen}
This section contains the main technical contribution of the paper, 
namely the proof that the generalized Medvedev logic 
$\mathsf{TLP}_{\mathbf{F}}$ is not finitely axiomatizable. 
In the general case, the proof relies on a structural analysis 
of $p$-morphisms from generated subframes of the topless products 
$\mathbf{F}^{m-}$. Lemma~\ref{lem:yi} provides an upper bound on branching degrees, 
which allows us to construct, for each frame $\mathbf{F}$, 
appropriate Chinese lantern frames. Combining this construction 
with classical results of Maksimova, Skvortsov, and Shehtman 
on Chinese lantern frames, we obtain the non-finite 
axiomatizability of $\mathsf{TLP}_{\mathbf{F}}$.

\subsection{The distributive lattice  case}
\label{subsec:chain}

We begin with the case of (non-trivial) finite distributive lattices, the following theorem shows that, the distributive lattice case is essentially the Medvedev logic. This follows from a recent characterization of Medvedev logic due to Grilletti \cite{Gri22}.

\begin{theorem}
Let $\mathbf{F}$ be a non-trivial finite distributive lattice. Then $\mathsf{TLP}_{\mathbf{F}} = \mathsf{Med}$.
\label{thm:chaincasemed}
\end{theorem}
\begin{proof}
By Lemma~\ref{thm:medmaxpift}, 
since $\mathsf{Med}$ is the largest generalized Medvedev logic, we already have $\mathsf{TLP}_{\mathbf{F}} \subseteq \mathsf{Med}$. Conversely, since $\mathbf{F}$ is a finite distributive lattice, each finite product $\mathbf{F}^n$ is again a finite distributive lattice. Hence $\mathbf{F}^{n-}$ is a frame obtained from a finite distributive lattice by removing its top element. By Grilletti's characterization of Medvedev logic as the logic of all such frames \cite{Gri22}, we have $\mathsf{Med}\subseteq \mathsf{Log}(\mathbf{F}^{n-})$ for every $n$. Taking the intersection over all $n$, we obtain $\mathsf{Med}\subseteq \bigcap_{n\in \omega} \mathsf{Log}(\mathbf{F}^{n-})=\mathsf{TLP}_{\mathbf{F}}.$
Therefore $\mathsf{TLP}_{\mathbf{F}}=\mathsf{Med}$.
\end{proof}
\begin{corollary}
For every finite chain $C(n)$ with $n\geq 2$, $\mathsf{TLP}_{\mathbf{C}(n)}=\mathsf{Med}$. Consequently, $\mathsf{TLP}_{\mathbf{C}(n)}$ is not finitely axiomatizable.
\label{cor:chaincasenick}	
\end{corollary}
\begin{proof}
Every finite chain is a finite distributive lattice, so the first claim follows from the Theorem~\ref{thm:chaincasemed}. The non-finite axiomatizability follows from the classical result of Maksimova, Skvortsov and Shehtman that $\mathsf{Med}$ is not finitely axiomatizable.
\end{proof}

\subsection{The general case}

Given integers $s$ and $n$, the \emph{Chinese lantern} 
$\mathbf{L}_{s,n}$ is the frame shown on the left in 
Figure~\ref{fig:chinese}. 
Let $\mathbf{T}_n$ be the poset consisting of a root and 
exactly $n$ immediate successors, and let $\mathbf{D}$ be 
the poset consisting of two incomparable points. 
Then $\mathbf{L}_{s,n} := \mathbf{T}_n \oplus (\mathbf{D} \times (s+1))$.
By removing the point $(m,1)$ from $\mathbf{L}_{s,n}$, 
we obtain the \emph{Chinese lantern} $\mathbf{L}'_{s,n,m}$ shown on the right in Figure~\ref{fig:chinese}.

\begin{figure}[htbp]
\centering
\begin{tikzpicture}[
scale=0.35,
dot/.style={circle, fill=black, minimum size=3.5pt, inner sep=0pt},
every label/.style={font=\scriptsize}
]

\begin{scope}[xshift=-7.5cm]
    \node[dot,label=355:{$(s+2,0)$}] (l1) at (0,1.5) {};
    \node[dot,label=left:{$(s+1,1)$}] (l2) at (-3.5,3.2) {};
    \node[dot] (l3) at (-2.25,3.2) {};
    \node[dot] (l4) at (2.25,3.2) {};
    \node[dot,label=right:{$(s+1,n)$}] (l5) at (3.5,3.2) {};

    \node[dot,label=left:{$(s,0)$}] (l6) at (-1.7,5) {};
    \node[dot,label=right:{$(s,1)$}] (l7) at (1.7,5) {};
    \node[dot,label=left:{$(s-1,0)$}] (l8) at (-1.7,6.5) {};
    \node[dot,label=right:{$(s-1,1)$}] (l9) at (1.7,6.5) {};

    \node[dot,label=left:{$(m+1,0)$}] (l10) at (-1.7,9) {};
    \node[dot,label=right:{$(m+1,1)$}] (l11) at (1.7,9) {};
    \node[dot,label=left:{$(m,0)$}] (l12) at (-1.7,10.5) {};
    \node[dot,label=right:{$(m,1)$}] (l13) at (1.7,10.5) {};
    \node[dot,label=left:{$(m-1,0)$}] (l14) at (-1.7,12) {};
    \node[dot,label=right:{$(m-1,1)$}] (l15) at (1.7,12) {};

    \node[dot,label=left:{$(1,0)$}] (l16) at (-1.7,14.5) {};
    \node[dot,label=right:{$(1,1)$}] (l17) at (1.7,14.5) {};
    \node[dot,label=left:{$(0,0)$}] (l18) at (-1.7,16) {};
    \node[dot,label=right:{$(0,1)$}] (l19) at (1.7,16) {};

    \draw[thin, dashed, dash pattern=on 1.6pt off 1pt] (l8) -- ++(1.02,0.6);
    \draw[thin, dashed, dash pattern=on 1.6pt off 1pt] (l9) -- ++(-1.02,0.6);
    \draw[thin, dashed, dash pattern=on 1.6pt off 1pt] (l10) -- ++(1.02,-0.6);
    \draw[thin, dashed, dash pattern=on 1.6pt off 1pt] (l11) -- ++(-1.02,-0.6);

    \draw[thin, dashed, dash pattern=on 1.6pt off 1pt] (l8) -- ++(0,0.6);
    \draw[thin, dashed, dash pattern=on 1.6pt off 1pt] (l9) -- ++(0,0.6);
    \draw[thin, dashed, dash pattern=on 1.6pt off 1pt] (l10) -- ++(0,-0.6);
    \draw[thin, dashed, dash pattern=on 1.6pt off 1pt] (l11) -- ++(0,-0.6);

    \draw[thin, dashed, dash pattern=on 1.6pt off 1pt] (l14) -- ++(1.02,0.6);
    \draw[thin, dashed, dash pattern=on 1.6pt off 1pt] (l15) -- ++(-1.02,0.6);
    \draw[thin, dashed, dash pattern=on 1.6pt off 1pt] (l16) -- ++(1.02,-0.6);
    \draw[thin, dashed, dash pattern=on 1.6pt off 1pt] (l17) -- ++(-1.02,-0.6);

    \draw[thin, dashed, dash pattern=on 1.6pt off 1pt] (l14) -- ++(0,0.6);
    \draw[thin, dashed, dash pattern=on 1.6pt off 1pt] (l15) -- ++(0,0.6);
    \draw[thin, dashed, dash pattern=on 1.6pt off 1pt] (l16) -- ++(0,-0.6);
    \draw[thin, dashed, dash pattern=on 1.6pt off 1pt] (l17) -- ++(0,-0.6);

    \fill (-1,3.2) circle (1.2pt);
    \fill (-0.75,3.2) circle (1.2pt);
    \fill (-0.5,3.2) circle (1.2pt);
    \fill (0.5,3.2) circle (1.2pt);
    \fill (0.75,3.2) circle (1.2pt);
    \fill (1,3.2) circle (1.2pt);

    \draw (l1) -- (l2);
    \draw (l1) -- (l3);
    \draw (l1) -- (l4);
    \draw (l1) -- (l5);

    \draw (l2) -- (l6);
    \draw (l2) -- (l7);
    \draw (l3) -- (l6);
    \draw (l3) -- (l7);
    \draw (l4) -- (l6);
    \draw (l4) -- (l7);
    \draw (l5) -- (l6);
    \draw (l5) -- (l7);

    \draw (l6) -- (l8);
    \draw (l6) -- (l9);
    \draw (l7) -- (l8);
    \draw (l7) -- (l9);

    \draw (l10) -- (l12);
    \draw (l10) -- (l13);
    \draw (l11) -- (l12);
    \draw (l11) -- (l13);
    \draw (l12) -- (l14);
    \draw (l12) -- (l15);
    \draw (l13) -- (l14);
    \draw (l13) -- (l15);

    \draw (l16) -- (l18);
    \draw (l16) -- (l19);
    \draw (l17) -- (l18);
    \draw (l17) -- (l19);
\end{scope}

\begin{scope}[xshift=8cm]
    \node[dot,label=355:{$(s+2,0)$}] (r1) at (0,1.5) {};
    \node[dot,label=left:{$(s+1,1)$}] (r2) at (-3.5,3.2) {};
    \node[dot] (r3) at (-2.25,3.2) {};
    \node[dot] (r4) at (2.25,3.2) {};
    \node[dot,label=right:{$(s+1,n)$}] (r5) at (3.5,3.2) {};

    \node[dot,label=left:{$(s,0)$}] (r6) at (-1.7,5) {};
    \node[dot,label=right:{$(s,1)$}] (r7) at (1.7,5) {};
    \node[dot,label=left:{$(s-1,0)$}] (r8) at (-1.7,6.5) {};
    \node[dot,label=right:{$(s-1,1)$}] (r9) at (1.7,6.5) {};

    \node[dot,label=left:{$(m+1,0)$}] (r10) at (-1.7,9) {};
    \node[dot,label=right:{$(m+1,1)$}] (r11) at (1.7,9) {};
    \node[dot,label={[label distance=0.6cm]180:{$(m,0)$}}] (r12) at (0,10.5) {};
    \node[dot,label=left:{$(m-1,0)$}] (r14) at (-1.7,12) {};
    \node[dot,label=right:{$(m-1,1)$}] (r15) at (1.7,12) {};

    \node[dot,label=left:{$(1,0)$}] (r16) at (-1.7,14.5) {};
    \node[dot,label=right:{$(1,1)$}] (r17) at (1.7,14.5) {};
    \node[dot,label=left:{$(0,0)$}] (r18) at (-1.7,16) {};
    \node[dot,label=right:{$(0,1)$}] (r19) at (1.7,16) {};

    \draw[thin, dashed, dash pattern=on 1.6pt off 1pt] (r8) -- ++(1.02,0.6);
    \draw[thin, dashed, dash pattern=on 1.6pt off 1pt] (r9) -- ++(-1.02,0.6);
    \draw[thin, dashed, dash pattern=on 1.6pt off 1pt] (r10) -- ++(1.02,-0.6);
    \draw[thin, dashed, dash pattern=on 1.6pt off 1pt] (r11) -- ++(-1.02,-0.6);

    \draw[thin, dashed, dash pattern=on 1.6pt off 1pt] (r8) -- ++(0,0.6);
    \draw[thin, dashed, dash pattern=on 1.6pt off 1pt] (r9) -- ++(0,0.6);
    \draw[thin, dashed, dash pattern=on 1.6pt off 1pt] (r10) -- ++(0,-0.6);
    \draw[thin, dashed, dash pattern=on 1.6pt off 1pt] (r11) -- ++(0,-0.6);

    \draw[thin, dashed, dash pattern=on 1.6pt off 1pt] (r14) -- ++(1.02,0.6);
    \draw[thin, dashed, dash pattern=on 1.6pt off 1pt] (r15) -- ++(-1.02,0.6);
    \draw[thin, dashed, dash pattern=on 1.6pt off 1pt] (r16) -- ++(1.02,-0.6);
    \draw[thin, dashed, dash pattern=on 1.6pt off 1pt] (r17) -- ++(-1.02,-0.6);

    \draw[thin, dashed, dash pattern=on 1.6pt off 1pt] (r14) -- ++(0,0.6);
    \draw[thin, dashed, dash pattern=on 1.6pt off 1pt] (r15) -- ++(0,0.6);
    \draw[thin, dashed, dash pattern=on 1.6pt off 1pt] (r16) -- ++(0,-0.6);
    \draw[thin, dashed, dash pattern=on 1.6pt off 1pt] (r17) -- ++(0,-0.6);

    \fill (-1,3.2) circle (1.2pt);
    \fill (-0.75,3.2) circle (1.2pt);
    \fill (-0.5,3.2) circle (1.2pt);
    \fill (0.5,3.2) circle (1.2pt);
    \fill (0.75,3.2) circle (1.2pt);
    \fill (1,3.2) circle (1.2pt);

    \draw (r1) -- (r2);
    \draw (r1) -- (r3);
    \draw (r1) -- (r4);
    \draw (r1) -- (r5);

    \draw (r2) -- (r6);
    \draw (r2) -- (r7);
    \draw (r3) -- (r6);
    \draw (r3) -- (r7);
    \draw (r4) -- (r6);
    \draw (r4) -- (r7);
    \draw (r5) -- (r6);
    \draw (r5) -- (r7);

    \draw (r6) -- (r8);
    \draw (r6) -- (r9);
    \draw (r7) -- (r8);
    \draw (r7) -- (r9);

    \draw (r10) -- (r12);
    \draw (r11) -- (r12);
    \draw (r12) -- (r14);
    \draw (r12) -- (r15);

    \draw (r16) -- (r18);
    \draw (r16) -- (r19);
    \draw (r17) -- (r18);
    \draw (r17) -- (r19);
\end{scope}

\end{tikzpicture}
\caption{The Chinese lanterns $\mathbf{L}_{s,n}$ (left) and $\mathbf{L}'_{s,n,m}$ (right).}
\label{fig:chinese}
\end{figure}
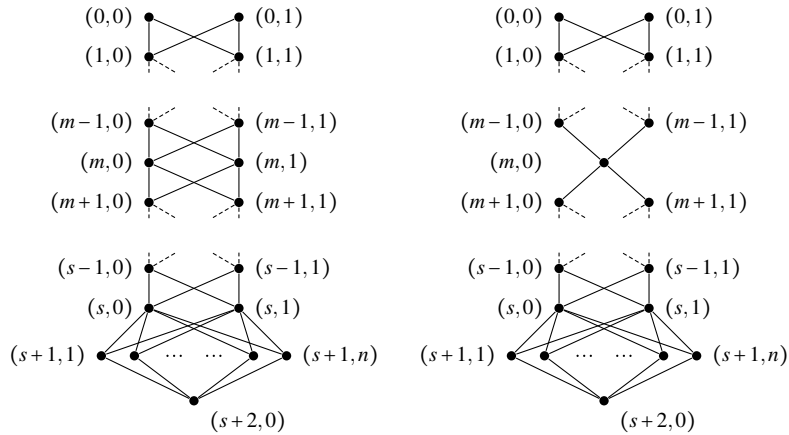

We define the $n$-th \emph{suspension} of a poset $\mathbf{T}$ inductively by $\mathbf{T}^{(0)} := \mathbf{T}$ and
$\mathbf{T}^{(n)} := \mathbf{T} \oplus (\mathbf{D}\times n)$. The following lemma appears as  \cite[Corollary 4]{MSS79}; a proof can also be found in \cite[Claim 9]{Fon07}.

\begin{lemma}[Maksimova, Skvortsov, \& Shehtman] 
Let $\mathbf{F}$ be a finite rooted frame with a top.
Then for every $n \ge 0$, the $n$-th suspension $\mathbf{F}^{(n)}$
is a $p$-morphic image of some Medvedev frame.
\label{lem:MSSsuspension} 
\end{lemma}
 
Lemma~\ref{lem:MSSsuspension} can be used to establish a connection 
between Chinese lantern frames and $\mathsf{TLP}_{\mathbf{F}}$-frames.

\begin{lemma}
	Each $\mathbf{L}'_{s, n, m}$ is a $\mathsf{TLP}_\mathbf{F}$-frame.
\label{lem:phiprimeispiftframe}
\end{lemma}

\begin{proof}
The downset $(m, 0){\downarrow}$ in $\mathbf{L}'_{s, n, m}$ is a finite rooted frame with a top $(m, 0)$. Moreover, $\mathbf{L}'_{s, n, m}$ is obtained as the $m$th suspension of $(m, 0){\downarrow}$. In other words, $\mathbf{L}'_{s, n, m}=\bigl((m, 0){\downarrow}\bigr)^{(m)}$. By Lemma \ref{lem:MSSsuspension}, $\mathbf{L}'_{s, n, m}$ is a $p$-morphic image of some Medvedev frame. Hence $\mathbf{L}'_{s,n,m}$ is a $\mathsf{Med}$-frame. Since, by Lemma~\ref{thm:medmaxpift}, 
every generalized Medvedev logic 
$\mathsf{TLP}_{\mathbf{F}}$ is contained in $\mathsf{Med}$, 
it follows that $\mathbf{L}'_{s,n,m}$ 
is a $\mathsf{TLP}_{\mathbf{F}}$-frame.
\end{proof}

The following propositions correspond to 
\cite[Lemma~9]{MSS79}; for a proof, see \cite[Proposition~11]{Fon07}.

\begin{proposition}[Maksimova, Skvortsov, \& Shehtman]
For any formula $\varphi$ with $s$ variables, there is an $m \leq s$ such that $\mathbf{L}_{s, n}\vDash \varphi$ if and only if $\mathbf{L}'_{s, n, m} \vDash \varphi$.
\label{prop:maksimova}	
\end{proposition}

Therefore, Chinese lantern frames provide an effective tool for detecting gaps between logics, in particular for determining whether they differ with respect to finite axiomatizability. 
We return to this idea in \S\ref{subsec:cheq}, where it is developed further in the context of $\mathsf{Cheq}$.

To prove Conjecture~\ref{con:generalcase}, 
we construct, for each finite rooted frame 
$\mathbf{F}$ with a top, 
a Chinese lantern frame that is not a 
$\mathsf{TLP}_{\mathbf{F}}$-frame. The lemma below provides the crucial construction. Beyond its role in the proof of the conjecture, it will later be used in our structural analysis of generalized Medvedev logics.

\begin{lemma}
Let $\mathbf{A}$ be a finite rooted frame such that no point of $\mathbf{A}$ has exactly one immediate successor. 
Let $\mathbf{F}$ be a non-singleton finite rooted frame with a top, 
and let $b := \max \{ b(u) : u \in \mathbf{F} \}$. If $\mathbf{A}$ is a $p$-morphic image of a generated subframe 
of some $\mathbf{F}^{m-}$, then $b(x) < b \cdot 2^{d(x)}$ for every $x \in \mathbf{A}$.\label{lem:yi}
\end{lemma}

\begin{proof}
Let $f : \mathbf{S} \to \mathbf{A}$ be a $p$-morphism, 
where $\mathbf{S}$ is a generated subframe of $\mathbf{F}^{m-}$. 
For $u=(u_1,\dots,u_m)\in\mathbf{S}$, 
let $t$ be the top of $\mathbf{F}$ and define $\#(u) := \lvert\{ i : u_i \neq t\}\rvert$. We prove by induction on $d(x)$ that for every $x\in\mathbf{A}$ 
there exists $u_x\in\mathbf{S}$ such that
$f(u_x)=x$ and $\#(u_x) < 2^{d(x)}$.

If $d(x)=1$, then $x$ is maximal. 
Choose $u$ with $f(u)=x$ and extend it to a maximal element 
$u_x\ge u$ in $\mathbf{S}$. 
Then $f(u_x)=x$ and $\#(u_x)=1<2$. Assume the claim holds for depth $d$. 
Let $d(x)=d+1$. 
Choose $u$ with $f(u)=x$. The frame $u{\uparrow}$ in $\mathbf{S}$ is also a generated subframe of $\mathbf{F}^{m-}$. The restriction of $f$ is a $p$-morphism from $u{\uparrow}$ to the finite rooted frame $x{\uparrow}$. By induction, for any proper successor $y$ of $x$, there is a $u_y\in u{\uparrow}$ such that $f(u_y)=y$ and $\#(u_y)< 2^{d(y)}$.

Since no point of $\mathbf{A}$ has exactly one immediate successor, 
$x$ has at least two distinct immediate successors $y,z$. Thus there exist $u_y,u_z\in u{\uparrow}$ with $f(u_y)=y$, $f(u_z)=z$ and $\#(u_y)<2^{d(y)}$, $\#(u_z)<2^{d(z)}$. Write $u=(u_1,\ldots,u_m)$, $u_y=(y_1,\ldots,y_m)$, and $u_z=(z_1,\ldots,z_m)$. For each $i$, let $x_i=t$ if $y_i=z_i=t$, and let $x_i=u_i$ otherwise. Put \(u_x=(x_1,\dots,x_m)\). Hence $u_x\in u{\uparrow}$ and $u_x< u_y$, $u_x < u_z$. Because $f$ is a $p$-morphism, $f(u_x)\in  x{\uparrow}$ and $f(u_x)\leq f(u_y)=y$, $f(u_x)\leq f(u_z)=z$. But $y$ and $z$ are distinct immediate successors of $x$, thus $f(u_x)=x$.

Observe that for each $1 \le i \le m$, $
x_i \neq t$ if and only if $
y_i \neq t \text{ or } z_i \neq t$. Hence, $\{ i : x_i \neq t \}
=\{ i : y_i \neq t \}\cup\{ i : z_i \neq t \}$. Therefore, $\#(u_x)
=
\lvert \{ i : x_i \neq t \} \rvert
\le
\lvert \{ i : y_i \neq t \} \rvert
+
\lvert \{ i : z_i \neq t \} \rvert
=
\#(u_y) + \#(u_z)$. By the induction hypothesis, $\#(u_y) < 2^{d(y)}$ and $\#(u_z) < 2^{d(z)}$. Since $d(y)\leq d$ and $d(z)\leq d$, we obtain $\#(u_x)<2^{d(y)} + 2^{d(z)}\leq 2^{d+1}=2^{d(x)}$. This completes the induction.

Since $b = \max\{ b(u) : u \in \mathbf{F}\}$, we have $b(x_i) \le b$ whenever $x_i \neq t$, and $b(x_i)=0$ if $x_i=t$. The immediate successors of \(u_x=(x_1,\dots,x_m)\) in \(\mathbf F^{m-}\)
are obtained by replacing exactly one coordinate \(x_i\) by an immediate
successor of \(x_i\) in \(\mathbf F\), except possibly when the resulting tuple is the deleted top element $(t,\ldots,t)$. Therefore, $b(u_x)\le \sum_{i=1}^m b(x_i)
\le
b \cdot \#(u_x)
<
b \cdot 2^{d(x)}$.

Fix $x \in \mathbf{A}$ and choose $u_x$ maximal among 
preimages of $x$ satisfying $\#(u_x) < 2^{d(x)}$. 
Such a choice is possible since $\mathbf{S}$ is finite. Let $y$ be an immediate successor of $x$. 
Because $f(u_x)=x \le y$ and $f$ is a $p$-morphism, 
there exists $u_y \in \mathbf{S}$ with $u_x \le u_y$ and $f(u_y)=y$. Since $u_x < u_y$, there exists an immediate successor 
$u'$ of $u_x$ with $u_x \le u' \le u_y$ and $x=f(u_x)\le f(u')\le f(u_y)=y$. If $f(u')=x$, then $u'$ would satisfy 
$f(u')=x$ and $\#(u')\le \#(u_x)<2^{d(x)}$, 
contradicting the maximality of $u_x$. 
Hence $f(u')=y$. Thus every immediate successor $y$ of $x$ 
lifts to an immediate successor $u'$ of $u_x$. Consequently, $b(x)\le b(u_x)$. Together with the bound $b(u_x)<b\cdot 2^{d(x)}$, this yields $b(x)<b\cdot 2^{d(x)}$ as required.
\end{proof}

\begin{corollary}
$\mathbf{L}_{s, b\cdot 2^{s+3}}$ is not a $\mathsf{TLP}_\mathbf{F}$-frame, where $b=\max\{b(u): u \in \mathbf{F}\}$.
\label{cor:phinotpiftframe}	
\end{corollary}
\begin{proof}
Let $x$ be the root of $\mathbf{L}_{s, b\cdot  2^{s+3}}$. Then $d(x)=s+3$ and $b(x)=b\cdot  2^{s+3}$. By Lemma \ref{lem:yi}, $\mathbf{L}_{s, b\cdot  2^{s+3}}$ is not a $p$-morphic image of any generated subframe of $\mathbf{F}^{m-}$.
  
Let $\chi$ be the Jankov-de Jongh formula of $\mathbf{L}_{s, b\cdot  2^{s+3}}$. For every $m$, $\mathbf{F}^{m-} \nvDash \chi$ if and only if $\mathbf{L}_{s,\, b\cdot 2^{s+3}}$ is a $p$-morphic image of a generated subframe of $\mathbf{F}^{m-}$. Since the latter is impossible, we conclude that $\mathbf{F}^{m-}\vDash \chi$ for all $m$. Thus $\chi \in \bigcap_{m\in\omega}\mathsf{Log}(\mathbf{F}^{m-})=\mathsf{TLP}_\mathbf{F}$. On the other hand, $\mathbf{L}_{s, b\cdot 2^{s+3}}\nvDash\chi$. Therefore, $\mathbf{L}_{s, b\cdot 2^{s+3}}$ is not a $\mathsf{TLP}_\mathbf{F}$-frame.                                      
\end{proof} 

\begin{theorem}
$\mathsf{TLP}_\mathbf{F}$ is not finitely axiomatizable.
\label{thm:generalcasenick}
\end{theorem}                                           
\begin{proof}
Suppose, towards a contradiction, that 
$\mathsf{TLP}_{\mathbf{F}}$ is finitely axiomatizable. Since a finite set of formulas $\varphi_1,\ldots,\varphi_n$ axiomatizes the same extension as the single formula $\varphi_1\wedge\cdots\wedge\varphi_n$, we may assume that $\mathsf{TLP}_{\mathbf{F}}$ is axiomatized by a single formula 
$\varphi(p_1,\dots,p_s)$, where $s$ is the number of propositional variables occurring in it. By Proposition~\ref{prop:maksimova}, 
there exists $m \le s$ such that $\mathbf{L}_{s,\, b\cdot 2^{s+3}} \vDash \varphi$ if and only if $\mathbf{L}'_{s,\, b\cdot 2^{s+3},\, m} \vDash \varphi$.

By Corollary~\ref{cor:phinotpiftframe}, $\mathbf{L}_{s,\, b\cdot 2^{s+3}} \nvDash \varphi$, whereas by Lemma~\ref{lem:phiprimeispiftframe}, $\mathbf{L}'_{s,\, b\cdot 2^{s+3},\, m} \vDash \varphi$. This contradiction shows that 
$\mathsf{TLP}_{\mathbf{F}}$ cannot be finitely axiomatized 
with $s$ variables. Since $s$ was arbitrary, 
the generalized Medvedev logics $\mathsf{TLP}_{\mathbf{F}}$ is not finitely axiomatizable.
\end{proof}

\section{Further results on generalized Medvedev logics}
This section collects two further consequences of the main results.
First, we strengthen Fontaine's theorem on $\mathsf{Med}$ by showing that, whenever
$\mathsf{Cheq}\subsetneq \mathsf{TLP}_{\mathbf{F}}$, the logic $\mathsf{TLP}_{\mathbf{F}}$
is not finitely axiomatizable even over $\mathsf{Cheq}$. Second, we turn to the global picture of generalized Medvedev logics under inclusion.
We show that there are countably many pairwise distinct generalized Medvedev logics and no least generalized Medvedev logic exists. 
\subsection{$\mathsf{TLP}_\mathbf{F}$ is not finitely axiomatizable over $\mathsf{Cheq}$}
\label{subsec:cheq}
In 2006, Fontaine showed that $\mathsf{Med}$ is not finitely 
axiomatizable over $\mathsf{Cheq}$ \cite[Theorem~9]{Fon06}. 
We extend this result to generalized Medvedev logics. 
The key observation is that all Chinese lantern frames are 
$\mathsf{Cheq}$-frames. Let $\mathbf{W}$ denote the \emph{bowtie}, the poset consisting 
of four elements: two incomparable minimal elements and two 
incomparable maximal elements, each maximal element lying above 
each minimal one; see Figure~\ref{fig:bowtie}. Our result relies on the following Bowtie Lemma, whose proof
can be found in~\mbox{\cite[Lemma~6.2.3]{Xiao24}}.

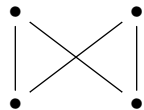
\begin{figure}[htbp]
$$\xymatrix@C=5mm{
\bullet\ar@{-}[d]\ar@{-}[rrd] & & \bullet\ar@{-}[d]\\
\bullet\ar@{-}[rru] & & \bullet\\
}$$
\caption{The bowtie $\mathbf{W}$.\label{fig:bowtie}}
\end{figure}

\begin{lemma}[Bowtie Lemma]	
If $\mathbf{F}$ is any finite rooted frame, then 
$\mathbf{F}\oplus\mathbf{W}$ is a $p$-morphic image of some 
$\mathbf{V}_n$, that is, the $n$-fold product of the two-fork.
\label{lem:bow-tie}
\end{lemma}

The Bowtie Lemma clarifies structural properties of 
$\mathsf{Cheq}$-frames and therefore provides a useful tool 
for the analysis of $\mathsf{Cheq}$-frame structures.

\begin{corollary}
The Chinese lantern frames $\mathbf{L}_{s,n}$ and 
$\mathbf{L}'_{s,n,m}$ are $\mathsf{Cheq}$-frames.
\label{lem:chineselanterncheq}
\end{corollary}
\begin{proof}
Every Chinese lantern $\mathbf{L}_{s,n}$ is a linear sum of a finite rooted frame and a bowtie. Hence, by the Bowtie Lemma~\ref{lem:bow-tie}, 
$\mathbf{L}_{s,n}$ is a $\mathsf{Cheq}$-frame.
For $m\ge 0$, the structure of $(m,0){\uparrow}$ in 
$\mathbf{L}'_{s,n,m}$ is determined by $m$.
If $m\ge 2$, it is a linear sum of a finite rooted frame and a bowtie.
If $m=1$, it is the frame $\mathbf{V}$.
If $m=0$, it is a singleton.
In each case, $(m,0){\uparrow}$ is a $p$-morphic image of some $V_{n_1}$
by the Bowtie Lemma~\ref{lem:bow-tie}. Because $(m,0){\downarrow}$ is a finite rooted frame with a top, it is also a $p$-morphic image of some Medvedev frame and thus a $p$-morphic image of some $V_{n_2}$.
By Claim~20 in~\cite{Fon07}, 
the vertical sum of the $p$-morphic images of $V_{n_1}$ and $V_{n_2}$
is a $\mathsf{Cheq}$-frame.
As $\mathbf{L}'_{s,n,m}$ is precisely the vertical sum of 
$(m,0){\downarrow}$ and $(m,0){\uparrow}$,
it follows that $\mathbf{L}'_{s,n,m}$ is a $\mathsf{Cheq}$-frame.
\end{proof}
\begin{theorem}
If $\mathsf{Cheq}\subsetneq \mathsf{TLP}_\mathbf{F}$, then $\mathsf{TLP}_\mathbf{F}$ is not finitely axiomatizable over $\mathsf{Cheq}$.  
\label{pro:genfontaine}		
\end{theorem}
\begin{proof}
Assume towards a contradiction that $\mathsf{TLP}_\mathbf{F}$ is finitely axiomatizable over $\mathsf{Cheq}$. We may then assume that there is a single formula $\varphi$ with $s$ propositional variables such that $\mathsf{TLP}_\mathbf{F}=\mathsf{Cheq}+\varphi$. By Proposition \ref{prop:maksimova}, we find $m \leq s$ such that
$\mathbf{L}_{s,b\cdot 2^{s+3}}\vDash \varphi$ if and only if $\mathbf{L}'_{s,b\cdot 2^{s+3},m}\vDash \varphi$.
By Corollary~\ref{lem:chineselanterncheq},
both $\mathbf{L}_{s,b\cdot 2^{s+3}}$ and 
$\mathbf{L}'_{s,b\cdot 2^{s+3},m}$ are $\mathsf{Cheq}$-frames. However, as in the proof of Theorem~\ref{thm:generalcasenick}, these two Chinese lanterns disagree on the validity of $\varphi$. Contradiction!
\end{proof}

\subsection{The order structure of generalized Medvedev logics}

Since no generalized Medvedev logic is finitely axiomatizable, 
it follows that many well-known logics cannot be generalized Medvedev logics. 
Together with the facts that $\mathsf{Med}$ is the greatest generalized Medvedev logic 
and that $\mathsf{Med} \subseteq \mathsf{KC}$, 
this yields an approximate description of the 
``geographic coordinates'' of generalized Medvedev logics in the lattice of intermediate logics.
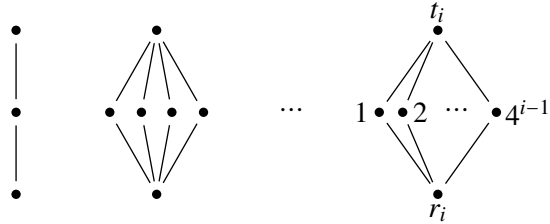
\begin{figure}[htbp]
\centering
\begin{tikzpicture}[
scale=0.31,
every node/.style={circle, fill=black, minimum size=3.5pt, inner sep=0pt},
edge/.style={line width=0.5pt, shorten >=3pt, shorten <=3pt}
]

\begin{scope}[xshift=0cm] 
    \node (a1) at (0,1) {};
    \node (a2) at (0,4.5) {};
    \node (a3) at (0,8) {};

    \draw[edge] (a1) -- (a2);
    \draw[edge] (a2) -- (a3);
\end{scope}

\begin{scope}[xshift=6cm] 
    \node (b1) at (0,1) {};
    \node (b2) at (-2,4.5) {};
    \node (b3) at (-0.65,4.5) {};
    \node (b4) at (0.65,4.5) {};
    \node (b5) at (2,4.5) {};
    \node (b6) at (0,8) {};
 
\node[draw=none, fill=none] at (5.75,4.5) {$\cdots$}; 

    \draw[edge] (b1) -- (b2);
    \draw[edge] (b1) -- (b3);
    \draw[edge] (b1) -- (b4);
    \draw[edge] (b1) -- (b5);
    \draw[edge] (b2) -- (b6);
    \draw[edge] (b3) -- (b6);
    \draw[edge] (b4) -- (b6);
    \draw[edge] (b5) -- (b6);
\end{scope}

\begin{scope}[xshift=18cm] 
    \node (c1)[label={[font=\normalsize]below:$r_i$}] at (0,1) {};
    \node (c2)[label={[font=\normalsize]left:$1$}] at (-2.5,4.5) {};
    \node (c3)[label={[font=\normalsize]right:$2$}] at (-1.5,4.5) {};
    \node (c4)[label={[font=\normalsize]right:$4^{i-1}$}] at (2.5,4.5) {};
    \node (c5)[label={[font=\normalsize]above:$t_i$}] at (0,8) {};
    
\node[draw=none, fill=none] at (0.8,4.5) {$\cdots$};
    \draw[edge] (c1) -- (c2);
    \draw[edge] (c1) -- (c3);
    \draw[edge] (c1) -- (c4);
    \draw[edge] (c2) -- (c5);
    \draw[edge] (c3) -- (c5);
    \draw[edge] (c4) -- (c5);
\end{scope}

\end{tikzpicture}

\caption{The members of the family $\{\mathbf{N}(i) : i\in\omega\}$.}
\label{fig:Di}
\end{figure}

\begin{theorem}
There are at least countably many pairwise distinct generalized Medvedev logics.
\label{thm:countablygen}
\end{theorem}
\begin{proof}
We construct a family $\{\mathbf{N}(i) : i\in \omega\}$ of finite rooted frames with its top. We begin by setting $\mathbf{N}(1)=\mathbf{C}(3)$. By Theorem~\ref{thm:chaincasemed}, we have $\mathsf{TLP}_{\mathbf{N}(1)}=\mathsf{Med}$. Next, define $\mathbf{N}(2)=\langle N_2,R_2\rangle$ by $N_2=\{r_2,1,2,3,4,t_2\}$ and $R_2=\{(r_2,x):x\in N_2\}
\cup
\{(x,t_2):x\in N_2\}
\cup
\{(x,x):x\in N_2\}$. Thus $r_2$ is the root of $\mathbf{N}(2)^{-}$, its depth is $2$, and its branching degree is $4$. Since $b(r_2)=2^{d(r_2)}$, by Lemma~\ref{lem:yi}, 
$\mathbf{N}(2)^{-}$ is not a $p$-morphic image of any generated subframe of 
$\mathbf{N}(1)^{m-}$ for any $m$. By Jankov-de Jongh Theorem, $\chi(\mathbf{N}(2)^{-})\in \mathsf{TLP}_{\mathbf{N}(1)}$ and $\chi(\mathbf{N}(2)^{-})\notin \mathsf{Log}(\mathbf{N}(2)^{-})$. Since $\mathsf{TLP}_{\mathbf{N}(1)}=\mathsf{Med}$ is the greatest generalized Medvedev logic, $\mathsf{TLP}_{\mathbf{N}(2)}$ is strictly contained in $\mathsf{TLP}_{\mathbf{N}(1)}$.

We now define, for $i\ge1$, $\mathbf{N}(i)= \langle N_i, R_i\rangle$ by $
N_i=\{r_i,1,2,\ldots,4^{\,i-1},t_i\}$ and $R_i=
\{(r_i,x):x\in N_i\}
\cup
\{(x,t_i):x\in N_i\}
\cup
\{(x,x):x\in N_i\}$. For $1\le i<j$, there is a natural $p$-morphism
$f_{ji}:\mathbf{N}(j)\to \mathbf{N}(i)$ given by $$f_{ji}(x)=
\begin{cases}
4^{\,i-1}, & \text{if } x\in\{4^{\,i-1}+1,\ldots,4^{\,j-1}\},\\
x, & \text{otherwise}.
\end{cases}$$ 

In addition, only $t_i$ of $\mathbf{N}(i)$ can be sent to  $t_j$ of $\mathbf{N}(j)$ via $f_{ji}$, thus for $1\leq i < j$, $\mathsf{TLP}_{\mathbf{N}(j)} \subseteq \mathsf{TLP}_{\mathbf{N}(i)}$. Moreover, the maximal branching degree of $\mathbf{N}(i)$
is $4^{\,i-1}$.
By Lemma~\ref{lem:yi},
$\mathbf{N}(j)^{-}$ is not a 
$\mathsf{TLP}_{\mathbf{N}(i)}$-frame when $i<j$,
since in $\mathbf{N}(j)^{-}$
the root has depth $2$ and branching degree $4^{\,j-1} \ge 2^{d(r_j)}\cdot 4^{\,i-1}$. Therefore, $\mathsf{TLP}_{\mathbf{N}(j)}
\subsetneq
\mathsf{TLP}_{\mathbf{N}(i)}$ whenever $1\le i<j$. We obtain a strictly descending chain $\mathsf{Med}
=\mathsf{TLP}_{\mathbf{N}(1)}
\supsetneq
\mathsf{TLP}_{\mathbf{N}(2)}
\supsetneq
\cdots
\supsetneq
\mathsf{TLP}_{\mathbf{N}(i)}
\supsetneq
\cdots$, which completes the proof.
\end{proof}

\begin{theorem}
	There is no least generalized Medvedev logic. 
\label{thm:noleastgen}
\end{theorem}
\begin{proof}
Assume towards a contradiction that there exists a least generalized Medvedev logic $\mathsf{B}$. Then there is a finite rooted frame $\mathbf{F}_0$ with a top such that $\mathsf{B}=\mathsf{TLP}_{\mathbf{F}_0}$ and $\mathsf{B}\subseteq \mathsf{TLP}_{\mathbf{F}}$ for every finite rooted frame $\mathbf{F}$ with a top. Let $b_0=\max\{b(u):u\in \mathbf{F}_0\}$. We construct a finite frame $\mathbf{M}=\langle M,R\rangle$ by $M=\{r_0,1,2,\ldots,4b_0,t_0\}$ and $R=
\{(r_0,x):x\in M\}
\cup
\{(x,t_0):x\in M\}
\cup
\{(x,x):x\in M\}$. In $\mathbf{M}^{-}$, the root $r_0$ has depth $d(r_0)=2$
and branching degree $4b_0$. Since $b(r_0)=4b_0
=
 b_0 \cdot 2^{d(r_0)}$, Lemma~\ref{lem:yi} implies that $\mathbf{M}^{-}$ is not a $p$-morphic image of any generated subframe of any $\mathbf{F_0}^{m-}$. Thus $\chi(\mathbf{M}^{-})\in \mathsf{TLP}_{\mathbf{F}_0}=\mathsf{B}$. On the other hand, $\mathsf{B} \subseteq \mathsf{TLP}_\mathbf{M}=\bigcap_{m\in \omega}\mathsf{Log}(\mathbf{M}^{m-})\subseteq \mathsf{Log}(\mathbf{M}^{-})$, so $\chi(\mathbf{M}^{-})\in \mathsf{Log}(\mathbf{M}^{-})$, a contradiction.
\end{proof}

\begin{corollary}
The intersection of all generalized Medvedev logics is no longer a generalized Medvedev logic. $\mathsf{IPC}$ is not a generalized Medvedev logic.
\end{corollary}  
\tikzset{
    dots on path/.style={
        postaction={
            decorate,
            decoration={
                markings,
                mark=at position 0.25 with {\fill[red] circle (1.3pt);},
                mark=at position 0.5 with {\fill[red] circle (1.3pt);},
                mark=at position 0.75 with {\fill[red] circle (1.3pt);}
            }
        }
    }
}

\begin{figure}[htbp]
\centering
\begin{tikzpicture}[scale=1,>=Stealth, node distance=1cm and 2cm]
    \node[circle, fill, inner sep=1pt, label=above:$\mathsf{CPC}$] (CPC) at (0,0) {};
    \node[circle, fill, inner sep=1pt, label=left:$\mathsf{KC}$] (KC) at (0,-0.5) {};
    \node[circle, fill=red, inner sep=1pt, label={[red]left:$\mathsf{Med}$}] (Med) at (0,-1) {};
    \node[red] (dots1) at (-1,-4) {$\vdots$};
    \node[red] (dots2) at (1,-4) {$\vdots$}; 
    \node[red] (dots4) at (-0.5,-3.9) {$\vdots$};

    \node[circle, fill, inner sep=1pt, label=below:$\mathsf{IPC}$] (IPC) at (0,-5.1) {};
    \node[red] (dots5) at (-1.1,-2.45) {\fontsize{8}{10}\selectfont Infinite};    
    \draw[very thick, dashed, red] (Med) to[out=225,in=90] (dots1);
    \draw[very thick, dashed, red] (Med) to[out=-45,in=90] (dots2);
    
    \draw[thick, red, dots on path] (Med) to[out=240,in=90]  (dots4);
    \draw[thick, dashed, red, dash pattern=on 1.5pt off 0.8pt] (dots4)-- ++(0, -0.45);

    \draw[->,thick,red] (dots5) to[out=254,in=190] (-0.54,-3.25);

    \draw[thick, dashed, bend right=82.5] (CPC) to (IPC);
    \draw[thick, dashed, bend left=82.5] (CPC) to (IPC);
\end{tikzpicture}
\vspace{10pt}
\caption{A schematic ``geography'' of generalized Medvedev logics. The red region indicates generalized Medvedev logics, all of which are not finitely axiomatizable.}
\end{figure}
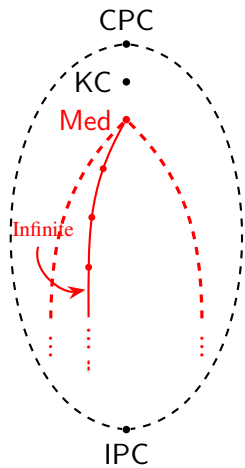

\section{Open problems and further directions}
Several directions for further research naturally arise from the present work. We already show that the class of generalized Medvedev logics has a non-trivial order structure: it contains a countable strictly descending chain, but has no least element. A related question is whether one can construct a countable antichain of generalized Medvedev logics under inclusion. It would also be interesting to understand when two finite rooted frames with a top give rise to the same generalized Medvedev logic, or to comparable ones. Given finite rooted frames with a top, $\mathbf{F}$ and $\mathbf{G}$, what features of
$\mathbf{F}$ and $\mathbf{G}$ determine whether the corresponding logics are equal, comparable, or incomparable? In particular, the present paper only gives some sufficient conditions for separating such logics, based on branching degree arguments. Another direction concerns the possible algebraic structure of this class. Although we have obtained a coarse ``geography'' of generalized Medvedev logics within the lattice of intermediate logics, a more refined structural description remains open. In particular, what algebraic or order-theoretic structure is formed by the class of all generalized Medvedev logics under inclusion? One may also consider a more general version of the construction, where instead of fixing a single finite rooted frame $\mathbf{F}$ with a top, one takes topless products whose frames are chosen from a class of finite rooted frames with a top. The present paper focuses on the single frame case because this is the direct analogue of the classical Medvedev construction, where one fixes the $2$-chain and considers the frames $\mathbf{C}(2)^{n-}$. The general class-based construction is therefore a natural further direction, but we leave this generalization for future work.

\bigskip
\noindent\textbf{Acknowledgements.}
The author would like to express his sincere gratitude to Nick Bezhanishvili 
for introducing him to the topic of this paper and for many 
inspiring discussions. In particular, he suggested the general research direction, 
conjectured Theorem~\ref{thm:pifkc}, and proposed the conjectures formulated in 
this paper as Conjecture~\ref{con:chaincase} and Conjecture~\ref{con:generalcase}.
The author is also especially grateful to Ga\"elle Fontaine for many helpful 
discussions. During the author's visit to ILLC in 2024, 
she generously shared her insights and discussed several aspects of the 
structure of generalized Medvedev logics. 
The author is grateful to the anonymous referees for their careful reading and
helpful comments. In particular, the author thanks the referees for pointing out
the relevance of Grilletti's characterization of Medvedev logic, and is especially
grateful to one referee for suggesting the stronger distributive lattice
version and its proof. These suggestions led to an improvement of
Subsection~\ref{subsec:chain}.
\nocite{*}
\bibliographystyle{eptcs}
\bibliography{generic}

\begin{thebibliography}{10}
\providecommand{\bibitemdeclare}[2]{}
\providecommand{\surnamestart}{}
\providecommand{\surnameend}{}
\providecommand{\urlprefix}{Available at }
\providecommand{\url}[1]{\texttt{#1}}
\providecommand{\href}[2]{\texttt{#2}}
\providecommand{\urlalt}[2]{\href{#1}{#2}}
\providecommand{\doi}[1]{doi:\urlalt{https://doi.org/#1}{#1}}
\providecommand{\eprint}[1]{arXiv:\urlalt{https://arxiv.org/abs/#1}{#1}}
\providecommand{\bibinfo}[2]{#2}

\bibitemdeclare{article}{vBBG03}
\bibitem{vBBG03}
\bibinfo{author}{Johan~\surnamestart van Benthem\surnameend},
  \bibinfo{author}{Guram~\surnamestart Bezhanishvili\surnameend} \&
  \bibinfo{author}{Mai~\surnamestart Gehrke\surnameend} (\bibinfo{year}{2003}):
  \emph{\bibinfo{title}{Euclidean Hierarchy in Modal Logic}}.
\newblock {\slshape \bibinfo{journal}{Studia Logica}}
  \bibinfo{volume}{75}(\bibinfo{number}{3}), pp. \bibinfo{pages}{327--344},
  \doi{10.1023/B:STUD.0000009564.00287.16}.

\bibitemdeclare{incollection}{BB22}
\bibitem{BB22}
\bibinfo{author}{Guram~\surnamestart Bezhanishvili\surnameend} \&
  \bibinfo{author}{Nick~\surnamestart Bezhanishvili\surnameend}
  (\bibinfo{year}{2022}): \emph{\bibinfo{title}{Jankov formulas and
  axiomatization techniques for intermediate logics}}.
\newblock In: {\slshape \bibinfo{booktitle}{V.A. Yankov on Non-Classical
  Logics, History and Philosophy of Mathematics}},
  \bibinfo{publisher}{Springer}, pp. \bibinfo{pages}{71--124},
  \doi{10.1007/978-3-031-06843-0_4}.

\bibitemdeclare{inproceedings}{Fon06}
\bibitem{Fon06}
\bibinfo{author}{Ga{\"e}lle~\surnamestart Fontaine\surnameend} (\bibinfo{year}{2006}):
  \emph{\bibinfo{title}{$\mathsf{ML}$ is not finitely axiomatizable over
  $\mathsf{Cheq}$}}.
\newblock In \bibinfo{editor}{Guido~\surnamestart Governatori\surnameend},
  \bibinfo{editor}{Ian~\surnamestart Hodkinson\surnameend} \&
  \bibinfo{editor}{Yde~\surnamestart Venema\surnameend}, editors: {\slshape
  \bibinfo{booktitle}{Advances in Modal Logic, 6th Conference, AiML 2006 on
  25-28 September, 2006}}, \bibinfo{publisher}{College Publications},
  \bibinfo{address}{Noosa, Queensland, Australia}, pp.
  \bibinfo{pages}{139--146}.
\newblock \urlprefix\url{https://dblp.org/rec/conf/aiml/Fontaine06}.

\bibitemdeclare{mastersthesis}{Fon07}
\bibitem{Fon07}
\bibinfo{author}{Ga{\"e}lle~\surnamestart Fontaine\surnameend} (\bibinfo{year}{2007}):
  \emph{\bibinfo{title}{Axiomatization of $\mathsf{ML}$ and $\mathsf{Cheq}$}}.
\newblock Master's thesis, \bibinfo{school}{Universiteit van Amsterdam}.
\newblock \urlprefix\url{https://eprints.illc.uva.nl/771/}.

\bibitemdeclare{inproceedings}{Gri22}
\bibitem{Gri22}
\bibinfo{author}{Gianluca~\surnamestart Grilletti\surnameend} (\bibinfo{year}{2022}):
  \emph{\bibinfo{title}{Medvedev logic is the logic of finite distributive
  lattices without top element}}.
\newblock In \bibinfo{editor}{David~\surnamestart Fern\'{a}ndez-Duque\surnameend},
  \bibinfo{editor}{Alessandra~\surnamestart Palmigiano\surnameend} \&
  \bibinfo{editor}{Sophie~\surnamestart Pinchinat\surnameend}, editors: {\slshape
  \bibinfo{booktitle}{Advances in Modal Logic, 14th Conference, AiML 2022 on
  22-25 August, 2022}}, \bibinfo{publisher}{College Publications},
  \bibinfo{address}{Rennes, Brittany, France}, pp. \bibinfo{pages}{451--466}.
\newblock \urlprefix\url{https://dblp.org/rec/conf/aiml/Grilletti22.html}.

\bibitemdeclare{article}{HLL15}
\bibitem{HLL15}
\bibinfo{author}{Joel David~\surnamestart Hamkins\surnameend},
  \bibinfo{author}{George~\surnamestart Leibman\surnameend} \&
  \bibinfo{author}{Benedikt~\surnamestart L{\"o}we\surnameend}
  (\bibinfo{year}{2015}): \emph{\bibinfo{title}{Structural connections between
  a forcing class and its modal logic}}.
\newblock {\slshape \bibinfo{journal}{Israel Journal of Mathematics}}
  \bibinfo{volume}{207}(\bibinfo{number}{2}), pp. \bibinfo{pages}{617--651},
  \doi{10.1007/s11856-015-1185-5}.
\newblock
  \urlprefix\url{https://link.springer.com/article/10.1007/s11856-015-1185-5}.

\bibitemdeclare{article}{HL08}
\bibitem{HL08}
\bibinfo{author}{Joel David~\surnamestart Hamkins\surnameend} \&
  \bibinfo{author}{Benedikt~\surnamestart L{\"o}we\surnameend}
  (\bibinfo{year}{2008}): \emph{\bibinfo{title}{The modal logic of forcing}}.
\newblock {\slshape \bibinfo{journal}{Transactions of the American Mathematical
  Society}} \bibinfo{volume}{360}(\bibinfo{number}{4}), pp.
  \bibinfo{pages}{1793--1817}, \doi{10.1090/S0002-9947-07-04297-3}.

\bibitemdeclare{article}{Kol32}
\bibitem{Kol32}
\bibinfo{author}{Andrey Nikolaevich~\surnamestart Kolmogorov\surnameend}
  (\bibinfo{year}{1932}): \emph{\bibinfo{title}{Zur {Deutung} der
  intuitionistischen {Logik}}}.
\newblock {\slshape \bibinfo{journal}{Mathematische Zeitschrift}}
  \bibinfo{volume}{35}, pp. \bibinfo{pages}{58--65}, \doi{10.1007/BF01186549}.

\bibitemdeclare{article}{Lit04}
\bibitem{Lit04}
\bibinfo{author}{Tadeusz~\surnamestart Litak\surnameend} (\bibinfo{year}{2004}):
  \emph{\bibinfo{title}{Some notes on the superintuitionistic logic of
  chequered subsets of $\mathbb{R}^{\infty}$}}.
\newblock {\slshape \bibinfo{journal}{Bulletin of the Section of Logic}}
  \bibinfo{volume}{33}(\bibinfo{number}{2}), pp. \bibinfo{pages}{81--86},
  \doi{10.48550/arXiv.1808.06393}.

\bibitemdeclare{inproceedings}{LX24}
\bibitem{LX24}
\bibinfo{author}{Benedikt~\surnamestart L{\"o}we\surnameend} \&
  \bibinfo{author}{Han~\surnamestart Xiao\surnameend} (\bibinfo{year}{2025}):
  \emph{\bibinfo{title}{Modal and intermediate logics of spiked {B}oolean
  algebras}}.
\newblock In \bibinfo{editor}{Cyriac~\surnamestart Aiswarya\surnameend},
  \bibinfo{editor}{Prabal Kumar~\surnamestart Sen\surnameend} \&
  \bibinfo{editor}{Shashi Mohan~\surnamestart Srivastava\surnameend}, editors:
  {\slshape \bibinfo{booktitle}{Logic and Its Applications, 11th Indian
  Conference, ICLA 2025 on 3-5 February, 2025}},
  \bibinfo{publisher}{Springer-Verlag}, \bibinfo{address}{Kolkata, India}, pp.
  \bibinfo{pages}{164--175}, \doi{10.1007/978-3-031-89610-1_12}.

\bibitemdeclare{article}{MSS79}
\bibitem{MSS79}
\bibinfo{author}{Larisa Lvovna~\surnamestart Maksimova\surnameend},
  \bibinfo{author}{Dmitry Pavlovich~\surnamestart Skvortsov\surnameend} \&
  \bibinfo{author}{Valentin Borisovich~\surnamestart Shehtman\surnameend}
  (\bibinfo{year}{1979}): \emph{\bibinfo{title}{Impossibility of finite
  axiomatization of {Medvedev}'s logic of finite problems}}.
\newblock {\slshape \bibinfo{journal}{Doklady Akademii Nauk SSSR}}
  \bibinfo{volume}{245}(\bibinfo{number}{5}), pp. \bibinfo{pages}{1051--1054}.
\newblock \urlprefix\url{https://www.mathnet.ru/eng/dan42643}.
\newblock \bibinfo{note}{English translation: Soviet Mathematics Doklady 20
  (1979), 394-398}.

\bibitemdeclare{article}{Med62}
\bibitem{Med62}
\bibinfo{author}{Yuri Tikhonovich~\surnamestart Medvedev\surnameend}
  (\bibinfo{year}{1962}): \emph{\bibinfo{title}{Finitive problems}}.
\newblock {\slshape \bibinfo{journal}{Doklady Akademii Nauk SSSR}}
  \bibinfo{volume}{142}(\bibinfo{number}{5}), pp. \bibinfo{pages}{1015--1018}.
\newblock \urlprefix\url{https://www.mathnet.ru/eng/dan26117}.
\newblock \bibinfo{note}{English translation: Soviet Mathematics Doklady 3
  (1962), 227-230}.

\bibitemdeclare{phdthesis}{Xiao24}
\bibitem{Xiao24}
\bibinfo{author}{Han~\surnamestart Xiao\surnameend} (\bibinfo{year}{2024}):
  \emph{\bibinfo{title}{Modal logics and intermediate logics motivated by an
  open problem on c.c.c.\ forcing}}.
\newblock Ph.D. thesis, \bibinfo{school}{Universit{\"a}t Hamburg}.
\newblock \urlprefix\url{https://ediss.sub.uni-hamburg.de/handle/ediss/11363}.

\end{thebibliography}
\end{document}